\documentclass[12pt,a4paper,leqno]{article}

\usepackage[pdftex]{graphics}

\usepackage{feynmf} 
\usepackage{scrtime}
\usepackage{epstopdf}
\usepackage{amsmath}
\usepackage{amsthm}
\usepackage{amssymb}
\usepackage{amsfonts}
\usepackage{mathrsfs,bbm}
\usepackage{verbatim}
\usepackage{enumerate}
\usepackage{mathrsfs}
\usepackage[all]{xy}
\usepackage{appendix}
\usepackage{array}
\usepackage{ifthen}
\usepackage{comment}

\numberwithin{equation}{section}

\theoremstyle{plain}
\newtheorem{theorem}{Theorem}[subsection]
\newtheorem{corollary}{Corollary}[subsection]

\newtheorem{proposition}{Proposition}[subsection]

\theoremstyle{definition}

\newtheorem{definition}{Definition}[subsection]

\theoremstyle{remark}
\newtheorem{example}{Example}[subsection]

\newcommand{\affichesolution}{0}
\newcounter{solution}
\ifthenelse{\equal{\affichesolution}{0}}
	   {\excludecomment{solution}}
	   {}

\DeclareFontFamily{OT1}{wncyi}{}
\DeclareFontShape{OT1}{wncyi}{m}{it}{
<5> <6> <7> <8> <9> gen * wncyi
<10> <10.95> <12> <14.4> <17.28> <20.74> <24.88> wncyi10
}{}
\DeclareSymbolFont{cyrletters}{OT1}{wncyi}{m}{it}
\DeclareSymbolFontAlphabet{\cyrmath}{cyrletters}
\DeclareMathSymbol{\rE}{\cyrmath}{cyrletters}{003}
\DeclareMathSymbol{\rD}{\cyrmath}{cyrletters}{068}
\DeclareMathSymbol{\rG}{\cyrmath}{cyrletters}{017}
\DeclareMathSymbol{\rI}{\cyrmath}{cyrletters}{073}
\DeclareMathSymbol{\rL}{\cyrmath}{cyrletters}{076}
\DeclareMathSymbol{\rZ}{\cyrmath}{cyrletters}{090}
\renewcommand{\phi}{\varphi}

\newcommand{\Hc}{\mathcal{H}}
\newcommand{\Oc}{\mathcal{O}}
\newcommand{\Ic}{\mathcal{I}}

\newcommand{\Dc}{\mathcal{D}}
\newcommand{\Bc}{\mathcal{B}}
\newcommand{\Mc}{\mathcal{M}}
\newcommand{\Nc}{\mathcal{N}}
\newcommand{\Lc}{\mathcal{L}}

\newcommand{\Sc}{\mathcal{S}}

\newcommand{\Pc}{\mathcal{P}}
\newcommand{\Ec}{\mathcal{E}}
\newcommand{\Fc}{\mathcal{F}}

\newcommand{\Rc}{\mathcal{R}}
\newcommand{\Gc}{\mathcal{G}}

\newcommand{\Cc}{\mathcal{C}}

\newcommand{\Vc}{\mathcal{V}}

\newcommand{\Ac}{\mathcal{A}}

\newcommand{\Db}{\mathbb{D}}
\newcommand{\A}{\mathbb{A}}
\newcommand{\Lb}{\mathbb{L}}

\newcommand{\R}{\mathbb{R}}
\newcommand{\uR}{\underline{\mathbb{R}}}

\newcommand{\Z}{\mathbb{Z}}

\newcommand{\N}{\mathbb{N}}

\newcommand{\ev}{\mathrm{ev}}

\newcommand{\Spec}{\mathrm{Spec}}
\newcommand{\orient}{\mathrm{or}}
\newcommand{\uSpec}{\underline{\mathrm{Spec}}}

\newcommand{\Ker}{\mathrm{Ker}}

\newcommand{\Jet}{\mathrm{Jet}}

\newcommand{\Sym}{\mathrm{Sym}}

\newcommand{\uSol}{\underline{\mathrm{Sol}}}

\newcommand{\Sol}{\mathrm{Sol}}
\newcommand{\Solc}{\mathcal{S}ol}

\newcommand{\gfk}{\mathfrak{g}}

\newcommand{\Hom}{\mathrm{Hom}}

\newcommand{\RHom}{\R\mathrm{Hom}}

\newcommand{\Der}{\mathrm{Der}}

\newcommand{\uHom}{\underline{\mathrm{Hom}}}

\newcommand{\uGamma}{\underline{\Gamma}}

\newcommand{\Homc}{\mathcal{H}om}

\newcommand{\Endc}{\mathcal{E}nd}
\newcommand{\RHomc}{\R\mathcal{H}om}
\newcommand{\Extc}{\mathcal{E}xt}

\newcommand{\End}{\mathrm{End}}

\newcommand{\Forget}{\mathrm{Forget}}

\newcommand{\id}{\mathrm{id}}

\newcommand{\Ber}{\mathrm{Ber}}
\newcommand{\limind}{{\underset{\longrightarrow}{\mathrm{lim}}\,}}

\newcommand{\Fib}{\mathrm{Fib}}
\newcommand{\Cof}{\mathrm{Cof}}
\newcommand{\Weq}{\mathrm{W}}

\newcommand{\DR}{\mathrm{DR}}

\newcommand{\Lotimes}{\overset{\Lb}{\otimes}}

\newcommand{\Sets}{\textsc{Sets}}
\newcommand{\SSets}{\textsc{SSets}}

\newcommand{\Alg}{\textsc{Alg}}

\newcommand{\Mod}{\textsc{Mod}}

\newcommand{\Sh}{\textsc{Sh}}

\newcommand{\Legos}{\textsc{Legos}}

\newcommand{\Var}{\textsc{Var}}

\newcommand{\Open}{\textsc{Open}}
\newcommand{\Rings}{\textsc{Rings}}

\newcommand{\Vect}{\textsc{Vect}}

\title{Homotopical poisson reduction of gauge theories}
\author{Fr\'ed\'eric Paugam}

\begin{document}
\maketitle
\abstract{The classical Poisson reduction of a given Lagrangian system
with (local) gauge symmetries has to be done before its quantization.
We propose here a coordinate free and self-contained mathematical presentation
of the covariant Batalin-Vilkovisky Poisson reduction of a general gauge theory.
It was explained in physical terms (DeWitt indices) in Henneaux and Teitelboim's book
\cite{Henneaux-Teitelboim}.
It was studied in coordinates using jet spaces by Barnich-Brandt-Henneaux
\cite{Barnich-Brandt-Henneaux-I}, Stasheff \cite{Stasheff1},
Fulp-Lada-Stasheff \cite{fulp-lada-stasheff}, among others.
The main idea of our approach is to use the functor of point approach to spaces of fields
to gain coordinate free geometrical insights on the spaces in play, and to focus on the notion of Noether
identities, that is a simple replacement of the notion of gauge symmetry, harder to handle
algebraically. Our main results are a precise formulation and understanding of the optimal finiteness hypothesis
necessary for the existence of a solution of the classical master equation,
and an interpretation of the Batalin-Vilkovisky construction in the setting of homotopical
geometry of non-linear partial differential equations.}

\tableofcontents
\newpage

\section*{Introduction}
This paper gives a self contained and coordinate free presentation of the Batalin-Vilkovisky
formalism for homotopical Poisson reduction of gauge theories, in the setting of algebraic
non-linear analysis, expanding on (and giving full proofs for) the very short presentation given in
\cite{Fred-histories-and-observables}, Section 4.  We have tried to be as self-contained as possible
so there may be some repetitions. We refer to loc. cit. for further references on the various subjects treated here.
We also refer to \cite{Fred-cours-m2-physique} for a more complete and detailed account of this theory
and of its applications in physics.

\section{Lagrangian variational problems}
\label{variational-problems}
For the reader's convenience, we recall shortly the formulation summed-up in \cite{Fred-histories-and-observables}
and fully described in \cite{Fred-cours-m2-physique} of general variational problems, and its grounding
on functorial geometry. This is certainly useful, but not strictly necessary to understand our final results.

\begin{definition}
\label{Lagrangian-variational-problem-definition}
A \emph{Lagrangian variational problem} is made of the following data:
\begin{enumerate}
\item A space $M$ called the parameter space for trajectories,
\item A space $C$ called the configuration space for trajectories,
\item A morphism $\pi:C\to M$ (often supposed to be surjective),
\item  A subspace $H\subset \Gamma(M,C)$ of the space of sections of $\pi$
$$\Gamma(M,C):=\{x:M\to C,\;\pi\circ x=\id\},$$
called the space of histories,
\item A functional (partial function with a domain of definition)
$$S:H\to R$$
on histories with values in a space $R$ usually given by $\R$ (or $\R[[\hbar]]$) called the action functional.
\end{enumerate}
The main object of classical physics is the space
$$T=\{x\in H|\,d_{x}S=0\}$$
of critical points of the action functional in histories.
\end{definition}
Recall from loc. cit. that the word \emph{space} of this definition means essentially a sheaf
$$X:\Legos^{op}\to \Sets$$
on a category $\Legos$ of geometrical building blocs equiped with a Grothendieck topology $\tau$,
also called a \emph{space modeled on $(\Legos,\tau)$}.
Spaces that are locally representable are called \emph{manifolds} or \emph{varieties}.
To present also higher gauge theory examples, one has to work with homotopical spaces,
but we will not do that here.

We refer to Deligne-Morgan's lectures \cite{notes-on-supersymmetry} 
and Manin's book \cite{Manin-supergeometry} for an introduction to super-varieties.
The reference \cite{Fred-cours-m2-physique} gives a complete account of this in the
functorial setting.
We will work without further comments with the following types of spaces:
\begin{enumerate}
\item Smooth spaces (also called diffeologies), modeled on the category $\Legos=\Open_{\Cc^\infty}$ of
open subsets of $\R^n$ for varying $n$ with smooth maps between them.
\item Smoothly algebraic spaces, modeled on the category $\Legos=\Alg_{\Cc^\infty}^{op}$ opposite
to the category of Lawvere's smooth algebras (see \cite{Moerdijk-Reyes}
and \cite{Fred-cours-m2-physique}), with its Zariski topology.
\item Smooth super-spaces, modeled on the category $\Legos=\Open_{\Cc^\infty}^s$
of smooth open subsets of the super affine space $\R^{n|m}$.
\item Smoothly algebraic super-spaces, modeled on the category $\Legos=\Alg^s_{\Cc^\infty}$
of smooth super-algebras, described in \cite{Fred-cours-m2-physique}.
\end{enumerate}
All these types of spaces are useful (and actually necessary) to describe differential calculus
on spaces of maps between smooth super-manifolds in a proper mathematical setting.
If $M$ and $C$ are two varieties in the above sense, the space of maps
$$\uHom(M,C):\Legos^{op}\to \Sets$$
is defined by
$$\uHom(M,C)(U)=\Hom(M\times U,C).$$
If $\pi:C\to M$ is a morphism of varieties, the space $\uGamma(M,C)$ is simply the corresponding
subspace in $\uHom(M,C)$.

\section{Algebraic analysis of partial differential equations}
\label{coordinate-free-PDEs}
In this section, we present the natural coordinate free approach to partial differential equations,
in the settings of $\Dc$-modules and $\Dc$-algebras.
We refer to Schapira's survey \cite{Schapira-triangulated-analysis} for an efficient introduction
to the general methods of linear algebraic analysis on varieties and to
Beilinson and Drinfeld's book \cite{Beilinson-Drinfeld-Chiral} for the non-linear setting.
 
This section expands on the article \cite{Fred-histories-and-observables},
giving more details and explanations. In particular, we use systematically the functor of point approach
to spaces of fields, as explained in Section \ref{variational-problems}
(see also \cite{Fred-cours-m2-physique} for a complete treatment) without further comments.
This means that spaces of superfunctions are treated essentially as usual spaces,
and functionals are defined as partially defined functions between (functors of points of) usual spaces.
We use Beilinson and Drinfeld's functorial approach \cite{Beilinson-Drinfeld-Chiral}
to non-linear partial differential equations, and we relate this approach to ours.
We are also inspired by Vinogradov \cite{Vinogradov} and Krasilshchik and
Verbovetsky \cite{Krasilshchik-Verbovetsky-1998}.

\subsection{$\Dc$-modules and linear partial differential equations}
\label{intro-D-modules}
We refer to Schneiders' review \cite{Schneiders} and Kashiwara's book \cite{Kashiwara-D-modules}
for an introduction to $\Dc$-modules. We just recall here basic results, that are necessary
for our treatment of non-linear partial differential equations in Section \ref{non-linear-PDEs}.

Let $M$ be a smooth variety of dimension $n$ and $\Dc$ be the algebra of
differential operators on $M$.
Recall that locally on $M$, one can write an operator $P\in \Dc$
as a finite sum
$$P=\sum_{\alpha}a_{\alpha}\partial^\alpha$$
with $a_{\alpha}\in \Oc_{M}$,
$$\partial=(\partial_{1},\dots,\partial_{n}):\Oc_{M}\to \Oc_{M}^n$$
the universal derivation and $\alpha$ some multi-indices.

To write down the equation $Pf=0$ with $f$ in an $\Oc_{M}$-module $\Sc$,
one needs to define the universal derivation $\partial:\Sc\to \Sc^n$. This is equivalent
to giving $\Sc$ the structure of a $\Dc$-module. The solution space
of the equation with values in $\Sc$ is then given by
$$\Solc_{P}(\Sc):=\{f\in \Sc,\;Pf=0\}.$$
Remark that
$$\Solc_{P}:\Mod(\Dc)\to \Vect_{\R_{M}}$$
is a functor that one can think of as representing the space of
solutions of $P$. Denote $\Mc_{P}$ the cokernel of the $\Dc$-linear map
$$\Dc\overset{.P}{\longrightarrow}\Dc$$
given by right multiplication by $P$.
Applying the functor $\Homc_{\Mod(\Dc)}(-,\Sc)$ to the exact sequence
$$\Dc\overset{.P}{\longrightarrow}\Dc\longrightarrow \Mc_{P}\to 0,$$
we get the exact sequence
$$0\to \Homc_{\Mod(\Dc)}(\Mc_{P},\Sc)\to \Sc\overset{P.}{\longrightarrow} \Sc,$$
which gives a natural isomorphism
$$\Solc_{P}(\Sc)=\Homc_{\Mod(\Dc)}(\Mc_{P},\Sc).$$
This means that the $\Dc$-module $\Mc_{P}$ represents
the solution space of $P$, so that the category of $\Dc$-modules is a convenient
setting for the functor of point approach to linear partial differential equations.

Remark that it is even better to consider the derived solution space
$$\R\Solc_{P}(\Sc):=\RHomc_{\Mod(\Dc)}(\Mc_{P},\Sc)$$
because it encodes also information on the inhomogeneous
equation
$$Pf=g.$$
Indeed, applying $\Homc_{\Dc}(-,\Sc)$ to the exact sequences
$$
0\to \Ic_{P}\to \Dc\to \Mc_{P}\to 0
$$
$$
0\to \Nc_{P}\to \Dc\to \Ic_{P}\to 0
$$
where $\Ic_{P}$ is the image of $P$ and $\Nc_{P}$ is its kernel,
one gets the exact sequences
$$
0\to \Homc_{\Dc}(\Mc_{P},\Sc)\to \Sc\to \Homc_{\Dc}(\Ic_{P},\Sc)\to \Extc^1_{\Dc}(\Mc,\Sc)\to 0\\
$$
$$
0\to \Homc_{\Dc}(\Ic_{P},\Sc)\to \Sc \to \Homc_{\Dc}(\Nc_{P},\Sc)\to \Extc_{\Dc}(\Ic_{P},\Sc)\to 0
$$
If $Pf=g$, then $QPf=0$ for $Q\in \Dc$ implies $Qg=0$. The second exact sequence
implies that this system, called the algebraic compatibility condition for the inhomogeneous
equation $Pf=g$ is represented by the $\Dc$-module $\Ic_{P}$, because
$$
\Homc_{\Dc}(\Ic_{P},\Sc)=\{g\in \Sc,\;Q.g=0,\;\forall Q\in \Nc_{P}\}.
$$
The first exact sequence shows that $\Extc^1_{\Dc}(\Mc,\Sc)$ are classes of vectors $f\in \Sc$
satisfying the algebraic compatibility conditions modulo those for which the system is truly compatible.
Moreover, for $k\geq 1$, one has
$$\Extc^1_{\Dc}(\Ic_{P},\Sc)\cong \Extc^{k+1}_{\Dc}(\Mc_{P},\Sc)$$
so that all the $\Extc^k_{\Dc}(\Mc_{P},\Sc)$ give interesting information about the differential
operator $P$.

Recall that the sub-algebra $\Dc$ of $\End_{\R}(\Oc)$,
is generated by the left multiplication by functions in $\Oc_{M}$
and by the derivation induced by vector fields in $\Theta_{M}$.
There is a natural right action of $\Dc$ on the $\Oc$-module $\Omega^n_{M}$
by
$$\omega.\partial=-L_{\partial}\omega$$
with $L_{\partial}$ the Lie derivative.

There is a tensor product in the category $\Mod(\Dc)$ given by
$$\Mc\otimes \Nc:=\Mc\otimes_{\Oc}\Nc.$$
The $\Dc$-module structure on the tensor product is given
on vector fields $\partial\in \Theta_{M}$ by Leibniz's rule
$$\partial(m\otimes n)=(\partial m)\otimes n+m\otimes (\partial n).$$
There is also an internal homomorphism object $\Homc(\Mc,\Nc)$ given by
the $\Oc$-module $\Homc_{\Oc}(\Mc,\Nc)$ equipped with the action
of derivations $\partial\in \Theta_{M}$ by
$$\partial(f)(m)=\partial(f(m))-f(\partial m).$$

An important system is given by the $\Dc$-module
of functions $\Oc$, that can be presented by the De Rham complex
$$
\Dc\otimes\Theta_{M}\to \Dc\to \Oc \to 0,
$$
meaning that $\Oc$, as a $\Dc$-module, is the quotient of $\Dc$ by the sub-$\Dc$-module generated by
vector fields. The family of generators $\partial_{i}$ of the kernel of $\Dc\to \Oc$
form a regular sequence, i.e., for every $k=1,\dots,n$, $\partial_{k}$ is not
a zero divisor in $\Dc/(\partial_{1},\dots,\partial_{k-1})$ (where $\partial_{-1}=0$ by convention).
This implies (see Lang \cite{Lang}, XXI \S 4 for more details on Koszul resolutions)
the following:
\begin{proposition}
The natural map
$$\Sym_{(\Mod_{dg}(\Dc),\otimes)}([\Dc\otimes\Theta_{M}\to \Dc])\longrightarrow \Oc$$
is a quasi-isomorphism of dg-$\Dc$-modules. The left hand side gives a free
resolution of $\Oc$ as a $\Dc$-module called the universal Spencer complex.
\end{proposition}

\begin{proposition}
The functor
$$\Mc\mapsto \Omega^n_{M}\otimes_{\Oc}\Mc$$
induces an equivalence of categories between the categories
$\Mod(\Dc)$ and $\Mod(\Dc^{op})$ of left and right $\Dc$-modules
whose quasi-inverse is
$$\Nc\mapsto \Homc_{\Oc_{M}}(\Omega^n_{M},\Nc).$$
The monoidal structure induced on $\Mod(\Dc^{op})$ by this equivalence is denoted $\otimes^!$.
\end{proposition}

\begin{definition}
Let $\Sc$ be a right $\Dc$-module.
The \emph{De Rham functor} with values in $\Sc$ is the functor
$$\DR_{\Sc}:\Mod(\Dc)\to \Vect_{\R_{M}}$$
that sends a left $\Dc$-module to
$$\DR_{\Sc}(\Mc):=\Sc\Lotimes_{\Dc}\Mc.$$
The \emph{De Rham functor} with values in $\Sc=\Omega^n_{M}$
is denoted $\DR$ and simply called the De Rham functor.
One also denotes $\DR^r_{\Sc}(\Mc)=\Mc\Lotimes_{\Dc}\Sc$
if $\Sc$ is a fixed left $\Dc$-module and $\Mc$ is a varying
right $\Dc$-module, and $\DR^r:=\DR^r_{\Oc}$.
\end{definition}

\begin{proposition}
\label{universal-de-Rham}
The natural map
$$
\begin{array}{ccc}
\Omega^n_{M}\otimes_{\Oc}\Dc&\to& \Omega^n_{M}\\
\omega\otimes Q&\mapsto&\omega(Q)
\end{array}
$$
extends to a $\Dc^{op}$-linear quasi-isomorphism
$$\Omega^*_{M}\otimes_{\Oc}\Dc[n]\overset{\sim}{\to}\Omega^n_{M}.$$
\end{proposition}
\begin{proof}
This follows from the fact that the above map is induced by tensoring
the Spencer complex by $\Omega^n_{M}$, and by the internal product isomorphism
$$
\begin{array}{ccc}
\Sym_{\Mod_{dg}(\Dc)}([\underset{-1}{\Dc\otimes\Theta_{M}}\to \underset{0}{\Dc\otimes\Oc}])\otimes \Omega^n_{M} & \longrightarrow & (\Omega^{*}_{M}\otimes\Dc[n],d)\\
X\otimes\omega								&\longmapsto & i_{X}\omega.
\end{array}
$$
\end{proof}

We will see that in the super setting, this proposition can be taken
as a definition of the right $\Dc$-modules of volume forms, called
Berezinians.

The $\Dc$-modules we will use are usually not $\Oc$-coherent but only
$\Dc$-coherent. The right duality to be used in the monoidal category $(\Mod(\Dc),\otimes)$ to
get a biduality statement for coherent $\Dc$-modules is thus not the internal duality
$\uHom_{\Oc}(\Mc,\Oc)$ but the derived dual $\Dc^{op}$-module
$$\Db(\Mc):=\RHomc_{\Dc}(\Mc,\Dc).$$
The non-derived dual works well for projective $\Dc$-modules, but most of the
$\Dc$-modules used in field theory are only coherent, so that one often uses
the derived duality operation. We now describe the relation (based on biduality)
between the De Rham and duality functors.
\begin{proposition}
\label{duality-de-Rham}
Let $\Sc$ be a coherent $\Dc^{op}$-module and $\Mc$ be a coherent $\Dc$-module.
There is a natural quasi-isomorphism
$$\R\Sol_{\Db(\Mc)}(\Sc):=\RHomc_{\Dc^{op}}(\Db(\Mc),\Sc)\cong\DR_{\Sc}(\Mc),$$
where $\Db(\Mc):=\RHom_{\Dc}(\Mc,\Dc)$ is the $\Dc^{op}$-module dual of $\Mc$.
\end{proposition}

The use of $\Dc$-duality will be problematic in the study of covariant
operations (like Lie bracket on local vector fields). We will come back
to this in Section \ref{local-operations}.

\subsection{Supervarieties and their Berezinians}
\label{Berezinian}
We refer to Penkov's article \cite{Penkov} for a complete study
of the Berezinian in the $\Dc$-module setting and to Deligne-Morgan's
lectures \cite{notes-on-supersymmetry} and Manin's book \cite{Manin-supergeometry}
for more details on super-varieties.
We also refer to \cite{Fred-cours-m2-physique} for
a treatment of smooth super-geometry making a systematic use of functors of points.

Let $M$ be a super-variety of dimension $n|m$ and denote $\Omega^1_M$ the
$\Oc_M$-module of differential forms on $M$
and $\Omega^*_M$ the super-$\Oc_M$-module of higher differential forms on $M$, defined as
the exterior (i.e., odd symmetric) power
$$\Omega^*_{M}:=\wedge^*\Omega^1_{M}:=\Sym_{\Mod(\Oc_M)}\Omega^1_{M}[1].$$
Remark that $\Omega^*_M$ is strickly speaking a $\Z/2$-bigraded $\R$-module, but we can
see it as a $\Z/2$-graded module because its diagonal $\Z/2$-grading
identifies with $\Sym_{\Mod(\Oc_M)}T\Omega^1_M$, where $T:\Mod(\Oc_M)\to \Mod(\Oc_M)$
is the grading exchange. Thus from now on, we consider $\Omega^*_M$ as a mere $\Z/2$-graded module.

The super version of Proposition \ref{universal-de-Rham}
can be taken as a definition of the Berezinian,
as a complex of $\Dc$-modules, up to quasi-isomorphism.
\begin{definition}
The \emph{Berezinian} of $M$ is defined
in the derived category of $\Dc_{M}$-modules by the formula
$$\Ber_{M}:=\Omega^*_{M}\otimes_{\Oc}\Dc[n].$$
The \emph{complex of integral forms} $I_{*,M}$ is defined by
$$I_{*,M}:=\RHom_{\Dc}(\Ber_{M},\Ber_{M}).$$
\end{definition}

The following proposition (see \cite{Penkov}, 1.6.3) gives a description of the Berezinian
as a $\Dc$-module.
\begin{proposition}
The Berezinian complex is concentraded in degree $0$, and equal there
to
$$\Ber_{M}:=\Extc^n_{\Dc}(\Oc,\Dc).$$
It is moreover projective of rank $1$ over $\Oc$.
\end{proposition}
\begin{proof}
This follows from the fact that
$$\wedge^*\Theta_{M}\otimes_{\Oc}\Dc[-n]\to \Oc$$
is a projective resolution such that
$$
\Ber_{M}:=\Omega^*_{M}\otimes_{\Oc}\Dc[n]=\R\Homc_{\Dc}(\wedge^*\Theta_{M}\otimes_{\Oc}\Dc[-n],\Dc)
$$
and this De Rham complex is exact
(Koszul resolution of a regular module)
except in degree zero where it is equal to
$$\Extc^n_{\Dc}(\Oc,\Dc).$$
\end{proof}

\begin{proposition}
Suppose $M$ is a super-variety of dimension $m|n$. The Berezinian is a locally
free $\Oc$-module of rank $1$ on $M$ with generator denoted
$D(dx_{1},\dots,dx_{m},d\theta_{1},\dots,d\theta_{n})$. It $f:M\to M$ is
an isomorphism of super-varieties (change of coordinate) with local tangent
map $D_{x}f$ described by the even matrix
$$D_{x}f=\left(\begin{array}{cc}A & B\\ C & D\end{array}\right)$$
acting on the real vector space
$$T_{x}M=(T_{x}M)^0\oplus (T_{x}M)^1,$$
the action of $D_{x}f$ on $D(dx_{1},\dots,dx_{m})$ is given by the Berezin determinant
$$\Ber(D_{x}f):=\det(A-BD^{-1}C)\det(D)^{-1}.$$
\end{proposition}
\begin{proof}
This is a classical result (see \cite{notes-on-supersymmetry} or \cite{Manin-supergeometry}).
\end{proof}

In the super-setting, the equivalence of left and right $\Dc$-modules
is given by the functor
$$\Mc\mapsto \Mc\otimes_{\Oc}\Ber_{M}$$
that twists by the Berezinian right $\Dc$-module,
which can be computed by using the definition
$$\Ber_{M}:=\Omega^*_{M}\otimes_{\Oc}\Dc[n]$$
and passing to degree $0$ cohomology.

A more explicit description of the complex of integral forms (up to
quasi-isomorphism) is given by
$$
I_{*,M}:=\RHomc_{\Dc}(\Ber_{M},\Ber_{M})
\cong
\Homc_{\Dc}(\Omega^*_{M}\otimes_{\Oc}\Dc[n],\Ber_{M})
$$
so that we get
$$
I_{*,M} \cong
\Homc_{\Oc}(\Omega^*_{M}[n],\Ber_{M})
\cong \Homc_{\Oc}(\Omega^*_{M}[n],\Oc)\otimes_{\Oc}\Ber_{M}
$$
and in particular $I_{n,M}\cong \Ber_{M}$.

Remark that Proposition \ref{universal-de-Rham} shows that
if $M$ is a non-super variety, then $\Ber_{M}$ is quasi-isomorphic
with $\Omega^n_{M}$, and this implies that
$$
I_{*,M}\cong \Homc_{\Oc}(\Omega^*_{M}[n],\Oc)\otimes_{\Oc}\Ber_{M}
\cong \wedge^*\Theta_{M}\otimes _{\Oc} \Omega^n_{M}[-n]\overset{i}{\longrightarrow}
\Omega^*_{M},
$$
where $i$ is the internal product homomorphism. This implies the isomorphism
$$
I_{*,M}\cong \Omega^*_{M},
$$
so that in the purely even case, integral forms essentially identify with
ordinary differential forms.

The main use of the module of Berezinians is given by its usefulness in the definition
of integration on super-varieties.
We refer to Manin \cite{Manin-supergeometry}, Chapter 4 for the following proposition.
\begin{proposition}
Let $M$ be a super-variety, with underlying variety $|M|$ and orientation sheaf $\orient_{|M|}$.
There is a natural integration map
$$\int_{M}[dt^{1}\dots dt^nd\theta^1\dots d\theta^q]:\Gamma_{c}(M,\Ber_{M}\otimes \orient_{|M|})\to \R$$
given in a local chart (i.e., an open subset $U\subset \R^{n|m}$) for $g=\sum_{I}g_{I}\theta^I\in \Oc$ by
$$\int_{U} [dt^{1}\dots dt^nd\theta^1\dots d\theta^q]g:=\int_{|U|} g_{1,\dots,1}(t)d^nt.$$
\end{proposition}

We finish by describing the inverse and direct image functors in the supergeometric
setting, following the presentation of Penkov in \cite{Penkov}.

Let $g:X\to Y$ be a morphism of supermanifolds. Recall that for $\Fc$ a sheaf
of $\Oc_Y$-modules on $Y$, we denote $g^{-1}\Fc$ the sheaf on $X$
defined by
$$g^{-1}\Fc(U):=\underset{g(U)\subset V}{\lim}\Fc(V).$$
The $(\Dc_{X},g^{-1}\Dc_{Y})$ module of relative
inverse differential operators is defined as
$$\Dc_{X\to Y}:=\Oc_{X}\otimes_{g^{-1}\Oc_{Y}}g^{-1}\Dc_{Y}.$$
The $(g^{-1}\Dc_{Y},\Dc_{X})$ module of relative
direct differential operators is defined as
$$
\Dc_{X\gets Y}:=
\Ber_{X}\otimes_{\Oc_{X}}\Dc_{X\to Y}\otimes_{g^{-1}\Oc_{Y}} (\Ber_{Y}^*).
$$

The inverse image functor of $\Dc$-modules is defined by
$$
g^*_{\Dc}(\_):=\Dc_{X\to Y}\Lotimes_{g^{-1}\Oc_{Y}}g^{-1}(\_):
D(\Dc_{Y})\to D(\Dc_{X}).
$$

If $g:X\hookrightarrow Y$ is a locally closed embedding,
the direct image functor is defined by
$$
g^\Dc_{*}(\_):=\Dc_{Y\gets X}\Lotimes_{\Dc_{X}}(\_):D(\Dc_{X})\to D(\Dc_{Y}).
$$
More generally, for any morphism $g:X\to Y$, one defines
$$g^\Dc_{*}(\_):=\R f_{*}(\_\Lotimes_{\Dc}\Dc_{X\to Y}).$$

\subsection{Differential algebras and non-linear partial differential equations}
\label{non-linear-PDEs}
In this section, we will use systematically the language of differential calculus
in symmetric monoidal categories,
and the functor of points approach to spaces of fields, described in \cite{Fred-histories-and-observables}.
We restrict our presentation to
polynomial partial differential equations with functional coefficients,
but our results also apply to smooth partial differential equations (for a full treatment, that would
be too long for this article, see \cite{Fred-cours-m2-physique}).
We specialize the situation to the symmetric monoidal category
$$(\Mod(\Dc_{M}),\otimes_{\Oc_{M}})$$
of left $\Dc_{M}$-module on a given (super-)variety $M$.
Recall that there is an equivalence
$$(\Mod(\Dc_{M}),\otimes)\to (\Mod(\Dc_{M}^{op}),\otimes^!)$$
given by tensoring with the $(\Dc,\Dc^{op})$-modules $\Ber_{M}$
and $\Ber_{M}^{-1}=\Homc_{\Oc}(\Ber_{M},\Oc)$. The unit objects
for the two monoidal structures are $\Oc$ and $\Ber_{M}$ respectively.
If $\Mc$ is a $\Dc$-module (resp. a $\Dc^{op}$-module), we denote
$\Mc^r:=\Mc\otimes\Ber_{M}$ (resp. $\Mc^\ell:=\Mc\otimes\Ber_{M}^{-1}$)
the corresponding $\Dc^{op}$-module (resp. $\Dc$-module).

Recall that if $P\in \Z[X]$ is a polynomial, one can study
the solution space
$$\uSol_{P=0}(A)=\{x\in A,\;P(x)=0\}$$
of $P$ with values in any commutative unital ring.
Indeed, in any such ring, one has a sum, a multiplication,
a zero and a unit that fulfill the necessary compatibilities
to be able to write down the polynomial. One can thus think
of the mathematical object given by the category of commutative
unital rings as solving the mathematical problem of giving a natural setting
for a coordinate free study of polynomial equations.
This solution space is representable, meaning that
there is a functorial isomorphism
$$\uSol_{P=0}(-)\cong \Hom_{\Rings_{cu}}(\Z[X]/(P),-).$$
This shows that the solution space of an equation essentially determine
the equation itself. Remark that the polynomial $P$
lives in the free algebra $\Z[X]$ on the given variable
that was used to write it.

Suppose now given the bundle $\pi_1:C=\R\times \R\to \R=M$ of smooth varieties.
We would like to study an algebraic non-linear partial differential equation
$$F(t,\partial_{t}^i x)=0$$
that applies to sections $x\in \Gamma(M,C)$, that are functions $x:\R\to \R$.
It is given by a polynomial $F(t,x_{i})\in \R[t,\{x_{i}\}_{i\geq 0}]$.
The solution space of such an equation can be studied with
values in any $\Oc$-algebra $\Ac$ equipped with an action of
the differentiation $\partial_{t}$ (that fulfills a Leibniz rule for multiplication),
the basic example being given
by the algebra $\Jet(\Oc_{C}):=\R[t,\{x_{i}\}_{i\geq 0}]$ above with the action
$\partial_{t}x_{i}=x_{i+1}$.
The solution space of the given partial differential equation is then given by the functor
$$\uSol_{\Dc,F=0}(\Ac):=\{x\in \Ac,\;F(t,\partial_{t}^i x)=0\}$$
defined on all $\Oc_{C}$-algebras equipped with an action of $\partial_{t}$.
To be more precise, we define the category of $\Dc$-algebras,
that solves the mathematical problem of finding a natural setting for
a coordinate free study of polynomial non-linear partial differential equations
with smooth super-function coefficients.

\begin{definition}
Let $M$ be a variety. A \emph{$\Dc_{M}$-algebra} is an algebra
$\Ac$ in the monoidal category of $\Dc_{M}$-modules.
More precisely, it is an $\Oc_{M}$-algebra equipped with
an action
$$\Theta_{M}\otimes  \Ac\to \Ac$$
of vector fields on $M$ such that the product in $\Ac$ fulfills Leibniz's rule
$$\partial(fg)=\partial(f)g+f\partial(g).$$
\end{definition}

Recall from \cite{Fred-cours-m2-physique} 
that one can extend the jet functor
to the category of smooth $\Dc$-algebras (and even to smooth super-algebras), to extend
the forthcoming results to the study of non-polynomial smooth partial differential equations.
The forgetful functor
$$\Forget:\Alg_{\Dc}\to \Alg_{\Oc}$$
has an adjoint (free $\Dc$-algebra on a given $\Oc$-algebra)
$$\Jet:\Alg_{\Oc}\to \Alg_{\Dc}$$
called the (infinite) jet functor. It fulfills the universal property that for every $\Dc$-algebra
$\Bc$, the natural map
$$\Hom_{\Alg_{\Oc}}(\Oc_{C},\Bc)\cong \Hom_{\Alg_{\Dc}}(\Jet(\Oc_{C}),\Bc)$$
induced by the natural map $\Oc_{C}\to\Jet(\Oc_{C})$ is a bijection.

Using the jet functor, one can show that the solution space
of the non-linear partial differential equation
$$F(t,\partial_{t}^i x)=0$$
of the above example is representable,
meaning that there is a natural isomorphism of functors on $\Dc$-algebras
$$\uSol_{\Dc,F=0}(-)\cong \Hom_{\Alg_{\Dc}}(\Jet(\Oc_{C})/(F),-)$$
where $(F)$ denotes the $\Dc$-ideal generated by $F$.
This shows that the jet functor plays the role of the polynomial
algebra in the differential algebraic setting.
If $\pi:C\to M$ is a bundle, we define
$$\Jet(C):=\uSpec(\Jet(\Oc_{C})).$$
One can summarize the above discussion by the following array:
\vskip0.5cm
\begin{tabular}{|c|c|c|}\hline
Equation	& Polynomial & Partial differential\\\hline
Formula 	& $P(x)=0$	 & $F(t,\partial^\alpha x)=0$\\
Naive variable & $x\in \R$ & $x\in \Hom(\R,\R)$\\
Algebraic structure & commutative unitary ring $A$   & $\Dc_{M}$-algebra $A$\\
Free structure & $P\in \R[x]$	& $F\in \Jet(\Oc_{C})$\\
Solution space & $\{x\in A,P(x)=0\}$ & $\{x\in A,F(t,\partial^\alpha x)=0\}$\\\hline
\end{tabular}

\begin{example}
If $\pi:C=\R^{n+m}\to \R^n=M$ is a trivial bundle of dimension $m+n$ over $M$ of dimension $n$,
with algebra of coordinates $\Oc_{C}:=\R[\underline{t},\underline{x}]$
for $\underline{t}=\{t^{i}\}_{i=1,\dots,n}$ and $\underline{x}=\{x^{j}\}_{j=1,\dots,m}$
given in multi-index notation, its jet algebra is
$$\Jet(\Oc_{C}):=\R[\underline{t},\underline{x}_{\alpha}]$$
where $\alpha\in \N^m$ is a multi-index representing the derivation index.
The $\Dc$-module structure is given by making $\frac{\partial}{\partial t^i}$ act
through the total derivative
$$D_{i}:=\frac{\partial}{\partial t^i}+\sum_{\alpha,k} x_{i\alpha}^k\frac{\partial}{\partial x_{\alpha}^k}$$
where $i\alpha$ denotes the multi-index $\alpha$ increased by one at the $i$-th coordinate.
For example, if $\pi:C=\R\times \R\to \R=M$, one gets
$$D_{1}=\frac{\partial}{\partial t}+x_{1}\frac{\partial}{\partial x}+x_{2}\frac{\partial}{\partial x_{1}}+\dots.$$
\end{example}

\begin{definition}
Let $\pi:C\to M$ be a bundle. A \emph{partial differential equation} on the space
$\Gamma(M,C)$ of sections of $\pi$ is given by a quotient
$\Dc_{M}$-algebra
$$p:\Jet(\Oc_{C})\twoheadrightarrow \Ac$$
of the jet algebra of the $\Oc_{M}$-algebra $\Oc_{C}$. Its \emph{local space of solutions}
is the $\Dc$-space whose points with values in $\Jet(\Oc_C)$-$\Dc$-algebras $\Bc$
are given by
$$
\uSol_{\Dc,(\Ac,p)}:=\{x\in \Bc|\;f(x)=0,\;\forall f\in \Ker(p)\}
$$
The \emph{non-local space of solutions} of the partial differential equation $(\Ac,p)$ is
the subspace of $\uGamma(M,C)$ given by
$$\uSol_{(\Ac,p)}:=\{x\in \uGamma(M,C)|\;(j_{\infty}x)^*L=0\textrm{ for all }L\in \Ker(p)\}$$
where $(j_{\infty}x)^*:\Jet(\Oc_{C})\to \Oc_{M}$ is (dual to) the Jet of $x$. Equivalently,
$x\in \uSol_{(\Ac,p)}$ if and only if there is a natural factorization
$$
\xymatrix@M+6pt{
\Jet(\Oc_{C})\ar[r]^(0.6){(j_{\infty}x)^*}\ar[dr]_{p} & \Oc_{M} \\
							     & \Ac\ar[u]}
$$
of the jet of $x$ through $p$.
\end{definition}

\subsection{Local functionals and local differential forms}
The natural functional invariant associated to a given $\Dc$-algebra
$\Ac$ is given by the De Rham complex
$$\DR(\Ac):=(I_{*,M}\otimes_{\Oc}\Dc[n])\Lotimes_{\Dc}\Ac$$
of its underlying $\Dc$-module with coefficient in the universal
complex of integral forms $I_{*,M}\otimes_{\Oc}\Dc[n]$,
and its cohomology $h^*(\DR(\Ac))$. We will denote
$$h(\Ac):=h^0(\DR(\Ac))=\Ber_{M}\otimes_{\Dc}\Ac$$
where $\Ber_{M}$ here denotes the Berezinian object (and not only the complex
concentrated in degree $0$). If $M$ is a non-super variety, one gets
$$
\DR(\Ac)=\Omega^n_{M}\Lotimes_{\Dc}\Ac\hspace{1cm}\textrm{and}\hspace{1cm}
h(\Ac)=\Omega^n_{M}\otimes_{\Dc}\Ac.
$$
The De Rham cohomology is given by the cohomology of the complex
$$\DR(\Ac)=I_{*,M}[n]\otimes_{\Oc_{M}}\Ac,$$
which gives
$$\DR(\Ac)=\wedge^*\Omega^1_{M}[n]\otimes_{\Oc_{M}}\Ac$$
in the non super case.

If $\Ac$ is a jet algebra of section of a bundle $\pi:C\to M$ with basis
a classical manifold, the De Rham complex identifies with a sub-complex of the usual De Rham
complex of $\wedge^*\Omega^1_{\Ac/\R}$ of $\Ac$ viewed as an ordinary ring.
One can think of classes in $h^*(\DR(\Ac))$ as defining a special class of (partially defined)
functionals on the space $\uGamma(M,C)$, by integration along singular homology cycles
with compact support.

\begin{definition}
Let $M$ be a super-variety of dimension $p|q$.
For every smooth simplex $\Delta_{n}$, we denote
$\Delta_{n|q}$ the super-simplex obtained by adjoining
$q$ odd variables to $\Delta_{n}$.
The \emph{singular homology of $M$ with compact support} is
defined as the homology $H_{*,c}(M)$ of the simplicial set
$$\Hom(\Delta_{\bullet|q},M)$$
of super-simplices with compact support condition on the body
and non-degeneracy condition on odd variables.
\end{definition}

Recall that, following \cite{Fred-cours-m2-physique},
a functional $f\in \Hom(\uGamma(M,C),\A^1)$ on a space of fields denotes
in general, by definition, only a partially defined function (with a well-chosen domain of
definition).
\begin{proposition}
Let $\pi:C\to M$ be a bundle and $\Ac$ be the $\Dc_{M}$-algebra
$\Jet(\Oc_{C})$.
There is a natural integration pairing
$$
\begin{array}{ccc}
H_{*,c}(M)  \times  h^{*-n}(\DR(\Ac)) & \to & \Hom(\uGamma(M,C),\A^1)\\
(\Sigma,\omega)					&\mapsto& [x\mapsto \int_{\Sigma}(j_{\infty}x)^*\omega]
\end{array}
$$
where $j_{\infty}x:M\to \Jet(C)$ is the taylor series of a given section $x$.
If $p:\Jet(\Oc_{C})\to \Ac$ is a given partial differential equation (such that $\Ac$ is $\Dc$-smooth) on
$\Gamma(M,C)$ one also gets
an integration pairing
$$
\begin{array}{ccc}
H_{*,c}(M)  \times  h^{*-n}(\DR(\Ac)) & \to & \Hom(\uSol_{(\Ac,p)},\A^1)\\
(\Sigma,\omega) &\mapsto& S_{\Sigma,\omega}:[x(t,u)\mapsto \int_{\Sigma}(j_{\infty}x)^*\omega].
\end{array}
$$
\end{proposition}
\begin{proof}
Remark that the values of the above pairing are given by partially defined functions,
with a domain of definition given by Lebesgue's domination condition to make
$t\mapsto \int_{\Sigma}(j_{\infty}x_{t})^*\omega$ a smooth function of $t$
if $x_{t}$ is a parametrized trajectory. The only point to check is that
the integral is independent of the chosen cohomology class. This follows from
the fact that the integral of a total divergence on a closed subspace is zero,
by Stokes' formula (the super case follows from the classical one).
\end{proof}

\begin{definition}
A functional $S_{\Sigma,\omega}:\uGamma(M,C)\to \A^1$
or $S_{\Sigma,\omega}:\uSol_{(\Ac,p)}\to \A^1$
obtained by the above constructed pairing is called a \emph{quasi-local functional}.
We denote
$$\Oc^{qloc}\subset \Oc:=\uHom(\uGamma(M,C),\A^1)$$
the space of quasi-local functionals.
\end{definition}

Remark that for $k\leq n$, the classes in $h^{*-k}(\DR(\Ac))$ are usually called
(higher) conservation laws for the partial differential equation $p:\Jet(\Oc_{C})\to \Ac$.

If $\Ac$ is a $\Dc$-algebra, i.e., an algebra in $(\Mod(\Dc),\otimes_{\Oc})$,
one defines, as usual, an $\Ac$-module of differential forms as an $\Ac$-module
$\Omega^1_{\Ac}$ in the monoidal category $(\Mod(\Dc),\otimes)$, equipped
with a ($\Dc$-linear) derivation $d:\Ac\to \Omega^1_{\Ac}$ such for every 
$\Ac$-module $\Mc$ in $(\Mod(\Dc),\otimes)$, the natural
map
$$
\Hom_{\Mod(\Ac)}(\Omega^1_\Ac,\Mc)\to \Der_{\Mod(\Ac)}(\Ac,\Mc)
$$
given by $f\mapsto f\circ d$ is a bijection.

Remark that the natural $\Oc$-linear map
$$\Omega^1_{\Ac/\Oc}\to \Omega^1_{\Ac}$$
is an isomorphism of $\Oc$-modules.
The $\Dc$-module structure on $\Omega^1_{\Ac}$ can be seen as an Ehresman
connection, i.e., a section of the natural projection
$$\Omega^1_{\Ac/\R}\to \Omega^1_{\Ac/M}.$$

\begin{example}
In the case of the jet space algebra $\Ac=\Jet(\Oc_{C})$ for $C=\R^{n+m}\to \R^n=M$,
a basis of $\Omega^1_{\Ac/M}$ compatible
with this section is given by the Cartan forms
$$\theta_{\alpha}^i=dx_{\alpha}^i-\sum_{j=1}^nx_{j\alpha}^idt^j.$$
The De Rham differential $d:\Ac\to \Omega^1_{\Ac}$ in the $\Dc$-algebra
setting and its De Rham cohomology (often denoted $d^V$ in the literature), can then be computed
by expressing the usual De Rham differential $d:\Ac\to \Omega^1_{\Ac/M}$
in the basis of Cartan forms.
\end{example}

As explained in Section \ref{intro-D-modules}, the right notion of finiteness and duality in
the monoidal category of $\Dc$-modules is not the $\Oc$-finite presentation and duality
but the $\Dc$-finite presentation and duality. This extends to the category of $\Ac$-modules
in $(\Mod(\Dc),\otimes)$. The following notion of smoothness differs from the usual one
(in general symmetric monoidal categories)
because we impose the $\Ac[\Dc]$-finite presentation, where $\Ac[\Dc]:=\Ac\otimes_{\Oc}\Dc$,
to have good duality properties.
\begin{definition}
The $\Dc$-algebra $\Ac$ is called \emph{$\Dc$-smooth} if $\Omega^1_{\Ac}$ is a projective
$\Ac$-module of finite $\Ac[\Dc]$-presentation in the category of $\Dc$-modules, and
$\Ac$ is a (geometrically) finitely generated $\Dc$-algebra, meaning that there exists
an $\Oc$-module $\Mc$ of finite type, and an ideal $\Ic\subset \Sym_{\Oc}(\Mc)$
and a surjection
$$\Jet(\Sym_{\Oc}(\Mc)/\Ic)\twoheadrightarrow \Ac.$$
\end{definition}

\begin{proposition}
If $\Ac=\Jet(\Oc_{C})$ for $\pi:C\to M$ a smooth map of varieties, then
$\Ac$ is $\Dc$-smooth and the $\Ac[\Dc]$-module $\Omega^1_{\Ac}$ is isomorphic to
$$\Omega^1_{\Ac}\cong \Omega^1_{C/M}\otimes_{\Oc_{C}}\Ac[\Dc].$$
\end{proposition}

In particular, if $\pi:C=\R\times M\to M$ is the trivial bundle with fiber coordinate $u$,
one gets the free $\Ac[\Dc]$-module of rank one
$$\Omega^1_{\Ac}\cong \Ac[\Dc]^{(\{du\})}$$
generated by the form $du$.

\subsection{Local vector fields and local operations}
\label{local-operations}
We refer to Beilinson-Drinfeld's book \cite{Beilinson-Drinfeld-Chiral}
for a complete and axiomatic study of general pseudo-tensor
categories. We will only present here the tools from this
theory needed to understand local functional calculus.
Local operations are new operations on $\Dc$-modules,
that induce ordinary multilinear operations on their De Rham cohomology.
We start by explaining the main motivation for introducing these new operations
when one does geometry with $\Dc$-algebras.

We now define the notion of local vector fields.
\begin{definition}
Let $\Ac$ be a smooth $\Dc$-algebra. The $\Ac^r[\Dc^{op}]$-module
of \emph{local vector fields} is defined by
$$\Theta_{\Ac}:=\Homc_{\Ac[\Dc]}(\Omega^1_{\Ac},\Ac[\Dc]),$$
where $\Ac^r[\Dc^{op}]:=\Ac^r\otimes_{\Ber_{M}}\Dc^{op}$
acts on the right though the isomorphism
$$\Ac^r[\Dc^{op}]\cong (\Ac\otimes_{\Oc}\Ber_{M})\otimes_{\Ber_{M}}\Dc^{op}\cong \Ac[\Dc^{op}].$$
\end{definition}

Remark now that in ordinary differential geometry, one way to define vector
fields on a variety $M$ is to take the $\Oc_{M}$-dual
$$\Theta_{M}:=\Homc_{\Oc_{M}}(\Omega^1_{M},\Oc_{M})$$
of the module of differential forms. The Lie bracket
$$[.,.]:\Theta_{M}\otimes \Theta_{M}\to \Theta_{M}$$
of two vector fields $X$ and $Y$ can then be defined from the universal derivation
$d:\Oc_{M}\to \Omega^1_{M}$, as the only
vector field $[X,Y]$ on $M$ such that for every function $f\in \Oc_{M}$, one
has the equality of derivations
$$[X,Y].f=X.i_{Y}(df)-Y.i_{X}(df).$$

In the case of a $\Dc$-algebra $\Ac$, this construction does not work directly because
the duality used to define local vector fields is not the $\Ac$-linear duality (because
it doesn't have good finiteness properties) but the $\Ac[\Dc]$-linear duality. This explains
why the Lie bracket of local vector fields and their action on $\Ac$ are new kinds of operations of the form
$$[.,.]:\Theta_{\Ac}\boxtimes \Theta_{\Ac}\to \Delta_{*}\Theta_{\Ac}$$
and
$$L:\Theta_{\Ac}\boxtimes \Ac\to \Delta_{*}\Ac$$
where $\Delta:M\to M\times M$ is the diagonal map and the box product is defined by
$$\Mc\boxtimes \Nc:=p_{1}^*\Mc\otimes p_{2}^*\Nc$$
for $p_{1},p_{2}:M\times M\to M$ the two projections.
One way to understand these construction is by looking at the natural injection
$$\Theta_{\Ac}\hookrightarrow \Homc_{\Dc}(\Ac,\Ac[\Dc])$$
given by sending $X:\Omega^1_{\Ac}\to \Ac[\Dc]$ to $X\circ d:\Ac\to \Ac[\Dc]$.
The theory of $\Dc$-modules tells us that the datum of this map is equivalent
to the datum of a $\Dc_{M\times M}^{op}$-linear map
$$L:\Theta_{\Ac}\boxtimes \Ac^r\to \Delta_{*}\Ac^r.$$
Similarly, the above formula
$$[X,Y].f=X.i_{Y}(df)-Y.i_{X}(df)$$
of ordinary differential geometry makes sense in local computations only if we see $\Theta_{\Ac}$
as contained in $\Homc_{\Dc}(\Ac,\Ac[\Dc])$ (and not in $\Homc_{\Dc}(\Ac,\Ac)$, contrary
to what is usually done), so that we must think of the bracket as a morphism of sheaves
$$\Theta_{\Ac}\to \Homc_{\Dc^{op}}(\Theta_{\Ac},\Theta_\Ac\otimes \Dc^{op})^r.$$
This is better formalized by a morphism of $\Dc_{M\times M}^{op}$-modules
$$[.,.]:\Theta_{\Ac}\boxtimes \Theta_{\Ac}\to \Delta_{*}\Theta_{\Ac}$$
as above.
Another way to understand these local operations is to make an analogy
with multilinear operations on $\Oc_{M}$-modules. Indeed, if $\Fc$, $\Gc$
and $\Hc$ are three quasi-coherent $\Oc_{M}$-modules, one has a natural
adjunction isomorphism
$$
\Homc_{\Oc_{M}}(\Fc\otimes \Gc,\Hc)\cong \Homc_{\Oc_{M}}(\Delta^*(\Fc\boxtimes\Gc),\Hc)
\cong \Homc_{\Oc_{M\times M}}(\Fc\boxtimes\Gc,\Delta_{*}\Hc),
$$
and local operations are given by a $\Dc$-linear version of the right part of
the above equality. It is better to work with this expression because of finiteness
properties of the $\Dc$-modules in play.

We now recall for the reader's convenience from Beilinson-Drinfeld \cite{Beilinson-Drinfeld-Chiral}
the basic properties of general local operations (extending straightforwardly their approach to the
case of a base supervariety $M$).

\begin{definition}
Let $(\Lc_{i})_{i\in I}$ be a finite family of $\Dc^{op}$-modules and
$\Mc$ be a $\Dc^{op}$-module. We define the \emph{space of $*$-operations}
$$P^*_{I}(\{\Lc_{i}\},\Mc):=\Hom_{\Dc_{X^I}}(\boxtimes \Lc_{i},\Delta_{*}^{(I)}\Mc)$$
where $\Delta^{(I)}:M\to M^I$ is the diagonal embedding and
$\boxtimes \Lc_{i}:=\otimes_{i\in I} p_{i}^*\Lc_{i}$ with $p_{i}:X^I\to X$ the natural
projections. The datum $(\Mod(\Dc^{op}),P_{I})$ is called the \emph{pseudo tensor
structure} on the category of $\Dc^{op}$-modules.
\end{definition}
One can define natural composition maps of pseudo-tensor operations, and
a pseudo-tensor structure is very similar to a tensor structure in many respect:
one has, under some finiteness hypothesis, good notions of internal homomorphisms and
internal $*$-operations.

The pseudo-tensor structure actually defines what one usually calls a colored operad
with colors in $\Mod(\Dc^{op})$. It is very important because it allows to easily manipulate
covariant objects like local vector fields or local Poisson brackets. The main idea of
Beilinson and Drinfeld's approach to geometry of $\Dc$-spaces is that functions and
differential forms usually multiply by using ordinary tensor product of $\Dc$-modules,
but that local vector fields and local differential operators multiply by using pseudo-tensor
operations. This gives the complete toolbox to do differential geometry on $\Dc$-spaces
in a way that is very similar to ordinary differential geometry.

Another way to explain the interest of $*$-operations, is that they induce ordinary operations
in De Rham cohomology. Since De Rham cohomology is the main
tool of local functional calculus (it gives an algebraic presentation of local functions,
differential forms and vector fields on the space $\uGamma(M,C)$ of trajectories of
a given field theory), we will make a systematic use of these operations. We refer to
Beilinson-Drinfeld \cite{Beilinson-Drinfeld-Chiral} for a proof of the following
result, that roughtly says that the De Rham functor can be extended to respect the
natural pseudo-tensor structures of its source and target categories. This means, in
simpler terms, that local, i.e., $*$-operations can be used to define usual operations on
quasi-local functionals.
\begin{proposition}
The central De Rham cohomology
functor $h:\Mod(\Dc^{op})\to \Mod(\R_{M})$
given by $h(\Mc):=h^0(\DR^r(\Mc)):=\Mc\otimes_{\Dc}\Oc$ induces a natural map
$$h:P^*_{I}(\{\Lc_{i}\},\Mc)\to \Hom(\otimes_{i}h(\Lc_{i}),h(\Mc))$$
from $*$-operations to multilinear operations. At the level of complexes, the choice
of a dg-$\Dc$-algebra resolution $\epsilon:\Pc\to \Oc$, that is flat as a
dg-$\Dc$-module (i.e., $\otimes_{\Dc}\Pc$ transforms acyclic complexes in
acyclic complexes), induces a natural morphism
$$\DR:P^*_{I}(\{\Lc_{i}\},\Mc)\longrightarrow \RHomc(\Lotimes_{i}\DR(\Lc_{i}),\DR(\Mc)).$$
\end{proposition}

\section{Gauge theories and the covariant phase space}
We are inspired, when discussing the Batalin-Vilkovisky (later called BV) formalism,
by a huge physical literature, starting with Peierls \cite{Peierls} and De Witt \cite{DeWitt}
for the covariant approach to quantum field theory,
and with  \cite{Henneaux-Teitelboim} and \cite{Fisch-Henneaux} as general references for the BV formalism.
More specifically, we also use Stasheff's work \cite{Stasheff2} and \cite{Stasheff1} as homotopical inspiration, and
\cite{fulp-lada-stasheff},  \cite{Barnich-BV-2010} and \cite{Cattaneo-Felder} for
explicit computations.

\subsection{A finite dimensional toy model}
\label{toy-model-BV}
In this section, we will do some new kind of differential geometry on spaces of the
form $X=\uSpec_{\Dc}(\Ac)$ given by spectra of $\Dc$-algebras,
that encode solution spaces of non-linear partial differential equations in a coordinate free fashion
(to be explained in the next section).
 
Before starting this general description, that is entailed of technicalities, we present a finite
dimensional analog, that can be used as a reference
to better understand the constructions done in the setting of $\Dc$-spaces.

Let $H$ be a finite dimensional smooth variety (analogous to the space of histories
$H\subset \uGamma(M,C)$ of a given Lagrangian variational problem)
and $S:H\to \R$ be a smooth function (analogous to an action functional on the space of histories).
Let $d:\Oc_{H}\to \Omega^1_{H}$ be the De Rham differential and
$$\Theta_{H}:=\Homc_{\Oc_{H}}(\Omega^1_{H},\Oc_{H}).$$
There is a natural biduality isomorphism
$$\Omega^1_{H}\cong \Homc_{\Oc_{H}}(\Omega^1_{H},\Oc_{H}).$$
Let $i_{dS}:\Theta_{H}\to \Oc_{H}$ be given by the insertion of vector
fields in the differential $dS\in \Omega^1_{H}$ of the given function $S:H\to \R$.

The claim is that there is a natural homotopical Poisson structure on the space
$T$ of critical points of $S:H\to \R$, defined by
$$T=\{x\in H,\; d_{x}S=0\}.$$
Define the algebra of functions $\Oc_{T}$ as the quotient of
$\Oc_{H}$ by the ideal $\Ic_{S}$ generated by the equations $i_{dS}(\vec v)=0$ for all
$\vec v\in \Theta_{H}$. Remark that $\Ic_{S}$ is the image of the insertion map
$i_{dS}:\Theta_{H}\to \Oc_{H}$ and is thus locally finitely generated by the image
of the basis vector fields $\vec x_{i}:=\frac{\partial}{\partial x_{i}}$ that correspond to the local coordinates
$x_{i}$ on $H$. Now let $\Nc_{S}$ be the kernel of $i_{dS}$. It describes
the relations between the generating equations $i_{dS}(\vec x_{i})=\frac{\partial S}{\partial x_{i}}$ of $\Ic_{S}$.

The differential graded $\Oc_{H}$-algebra
$$\Oc_{P}:=\Sym_{dg}([\underset{-1}{\Theta_{H}}\overset{i_{dS}}{\longrightarrow}\underset{0}{\Oc_{H}}])$$
is isomorphic, as a graded algebra, to the algebra of multi-vectors
$$\wedge^*_{\Oc_{H}}\Theta_{H}.$$
This graded algebra is equipped with an odd, so called Schouten bracket, given by extending
the Lie derivative
$$
L:\Theta_{H}\otimes \Oc_{H}\to \Oc_{H}
$$
and Lie bracket
$$
[.,.]:\Theta_{H}\otimes \Theta_{H}\to \Theta_{H}
$$
by Leibniz's rule.

\begin{proposition}
The Schouten bracket is compatible with the insertion map $i_{dS}$ and makes
$\Oc_{P}$ a differential graded odd Poisson algebra. The Lie bracket on $\Theta_{H}$
induces a Lie bracket on $\Nc_{S}$.
\end{proposition}

Now, we will compute, following Tate \cite{Tate4}, a cofibrant resolution
$$\Bc\overset{\sim}{\to} \Oc_H/\Ic_S=:\Oc_T$$
of the algebra of functions on the critical space, as an $\Oc_H$-algebra. This proceeds
by adding inductively to the algebra $\Oc_P$ higher degree generators to annihilate
its cohomology. More precisely, if $H^1(\Oc_P,i_{dS})\neq 0$, we choose a submodule
$$\gfk_S^{0}\subset \Ker(i_{dS})\subset \Theta_H$$
such that
$$H^1\left(\Sym_{dg-\Oc_H}([\underset{-2}{\gfk_S^0}\to \underset{-1}{\Theta_H}\to \underset{0}{\Oc_H}])\right)=0.$$
And we apply the same procedure by adding a chosen $\gfk_S^1$ to kill the $H^2$ of
the algebra
$$\Sym_{dg-\Oc_H}([\underset{-2}{\gfk_S^0}\to \underset{-1}{\Theta_H}\to \underset{0}{\Oc_H}]),$$
and so on...
The graded module $\gfk_S$ is the finite dimensional analog of the space of gauge (and higher gauge)
symmetries.
Suppose that $\gfk_S$ is bounded with projective components of finite rank.

\begin{definition}
The \emph{finite dimensional BV algebra} associated to $S:H\to \R$ is the bigraded $\Oc_{H}$-algebra
$$
\Oc_{BV}:=
\Sym_{bigrad}\left(\left[
\begin{array}{ccccc}
\gfk_{S}[2] &\oplus & \Theta_{H}[1] &\oplus & \Oc_{H}\\
&&&&\oplus\\
&&&& ^t\gfk_{S}^{*}[-1]
\end{array}\right]\right),
$$
where $^t\gfk_{S}^*$ is the $\Oc_{H}$-dual of the graded module $\gfk_{S}$ transposed
to become a vertical ascending graded module.
\end{definition}

The main theorem of the Batalin-Vilkovisky formalism, that is the aim of this
section, is the following:
\begin{theorem}
There exists a non-trivial extension
$$S_{cm}=S_{0}+\sum_{i\geq 1}S_{i}\in \Oc_{BV}$$
of the classical function $S$ that fulfills the \emph{classical master equation}
$$\{S_{cm},S_{cm}\}=0.$$
The differential $D=\{S_{cm},.\}$ gives $\Oc_{BV}$ the structure of a differential graded odd
Poisson algebra.
\end{theorem}

The corresponding derived space $\R\uSpec(\Oc_{BV},D)$
can be thought as a kind of derived Poisson reduction
$$\R\uSpec(\Oc_{BV},D)\cong \R\uSpec(\Oc_{H}/\Ic_{H})\underset{\R}{/}\Nc_{S}$$
that corresponds to taking the quotient of the homotopical critical space of $S$
(cofibrant replacement of $\Oc_{T}=\Oc_{H}/\Ic_{H})$) by the foliation induced by
the Noether relations $\Nc_{S}$.

\begin{example}
To be even more explicit, let us treat a simple example of the above construction.
Let $H=\R^2$ be equipped with the polynomial function algebra $\Oc_{H}=\R[x,y]$.
The differential one forms on $H$ are given by the free $\Oc_{H}$-modules
$$\Omega^1_{H}=\R[x,y]^{(dx,dy)}.$$
Let $S\in \Oc_{H}$ be the function $F(x,y)=\frac{x^2}{2}$. One then has
$dS=xdx$. The module $\Theta_{H}$ of vector fields is the free module
$$\Theta_{H}=\R[x,y]^{\left(\frac{\partial}{\partial x},\frac{\partial}{\partial y}\right)}$$
and the insertion map is given by the $\R[x,y]$-module morphism
$$
\begin{array}{cccc}
i_{dS}: 	& \Theta_{H} & \to & \Oc_{H}\\
		& \vec v		&\mapsto & \langle dS,\vec v\rangle,
\end{array}
$$
and, in particular,
$i_{dS}\left(\frac{\partial}{\partial x}\right)=x$, and
$i_{dS}\left(\frac{\partial}{\partial y}\right)=0$.
The image of the insertion map is given by the ideal $\Ic_{S}=(x)$ in $\Oc_{H}=\R[x,y]$.
The kernel $\Nc_{S}$ of the insertion map is the free submodule
$$
\Nc_{S}=\gfk_{S}^0=
\R[x,y]^{(\frac{\partial}{\partial y})}
$$
of $\Theta_{H}$.
We apply the above inductive construction of a Koszul-Tate resolution to get
a graded module $\gfk_S$.
The corresponding cofibrant resolution is given by a quasi-isomorphism
$$
\Bc_S \overset{\sim}{\to} \Oc_{H}/\Ic_{S}.
$$
The obtained algebra $\Bc$ is called the Koszul-Tate resolution of $\Oc_{H}/\Ic_{S}$.
Now the bigraded BV algebra $\Oc_{BV}$ is given by
$$
\Oc_{BV}  := 
\Sym_{bigrad}\left(\left[
\begin{array}{ccccc}
\gfk_{S}[2] &\oplus & \Theta_{H}[1] &\oplus & \Oc_{H}\\
&&&&\oplus\\
&&&& ^t\gfk_{S}^{*}[-1]
\end{array}\right]\right).
$$
The graded version is the $\Oc_{H}$-algebra
$$
\Oc_{BV}=
\Sym^*(\gfk_{S})\otimes
\wedge^*\Theta_{H}\otimes \wedge^*\gfk_{S}^*
$$
on $4=1+2+1$ graded variables.
The BV formalism gives a way to combine the Koszul-Tate
differential with the Chevalley-Eilenberg differential of the action of $\Nc_S$ on
$\Oc_T$, by constructing an $S_{cm}\in \Oc_{BV}$
such that some components of the bracket $\{S_{cm},.\}$
induce both differentials on the corresponding generators
of the BV algebra.
\end{example}

The aim of this section is to generalize the above construction to local variational problems,
where $H\subset \uGamma(M,C)$ is a space of histories (subspace of the space of
sections of a bundle $\pi:C\to M$) and $S:H\to \R$ is given by the integration
of a Lagrangian density.
The main difficulties that we will encounter and overcome trickily in this generalization are that:
\begin{enumerate}
\item The $\Dc$-module $\Dc\otimes_{\Oc}\Dc$ is not $\Dc$-coherent, so that a projective
resolution of $\Nc_{S}$ will not be dualizable in practical cases. This will impose us
to use finer resolutions of the algebra $\Oc_{T}$ of functions on the critical space.
\item Taking the bracket between two \emph{densities} of vector fields on the space $H$ of histories
is not an $\Oc_{M}$-bilinear operation but a new kind of operation, called a locally
bilinear operation and described in Section \ref{local-operations}.
\item The ring $\Dc$ is not commutative. This will be overcome by using the equivalence
between $\Dc$-modules and $\Dc^{op}$-modules given by tensoring by $\Ber_{M}$.
\end{enumerate}
Out of the above technical points, the rest of the constructions of this section
are completely parallel to what we did on the finite dimensional toy model.

\subsection{General gauge theories}
\begin{proposition}
Let $M$ be a super-variety and $\Ac$ be a smooth $\Dc$-algebra.
There is a natural isomorphism, called the local interior product
$$
\begin{array}{cccc}
i: 	& h(\Omega^1_{\Ac}):=\Ber_{M}\otimes_{\Dc}\Omega^1_{\Ac} & \longrightarrow &
\Homc_{\Ac[\Dc]}(\Theta_{\Ac}^\ell,\Ac)\\
	& \omega	& \longmapsto & [X\mapsto i_{X}\omega].
\end{array}
$$
\end{proposition}
\begin{proof}
By definition, one has
$$\Theta_{\Ac}:=\Homc_{\Ac[\Dc]}(\Omega^1_{\Ac},\Ac[\Dc])$$
and since $\Ac$ is $\Dc$-smooth, the biduality map
$$\Omega^1_{\Ac}\to \Homc_{\Ac^r[\Dc^{op}]}(\Theta_{\Ac},\Ac^r[\Dc^{op}])$$
is an isomorphism. Tensoring this map with $\Ber_{M}$ over $\Dc$ gives the desired result.
\end{proof}

\begin{definition}
If $\Ac$ is a smooth $\Dc$-algebra and $\omega\in h(\Omega^1_{\Ac})$, the $\Ac[\Dc]$ linear
map
$$i_{\omega}:\Theta_{\Ac}^\ell\to \Ac$$
is called the \emph{insertion map}. Its kernel $\Nc_{\omega}$ is called the $\Ac[\Dc]$-module
of \emph{Noether identities} and its image $\Ic_{\omega}$ is called the \emph{Euler-Lagrange ideal}.
\end{definition}

If $\Ac=\Jet(\Oc_{C})$ for $\pi:C\to M$ a bundle and $\omega=dS$,
the Euler-Lagrange ideal $\Ic_{dS}$ is locally generated as an $\Ac[\Dc]$-module
by the image of the local basis of vector fields in
$$\Theta_{\Ac}^\ell\cong \Ac[\Dc]\otimes_{\Oc_{C}}\Theta_{C/M}.$$
If $M$ is of dimension $n$ and the relative dimension of $C$ over $M$ is
$m$, this gives $n$ equations (indexed by $i=1,\dots,n$, one for each generator
of $\Theta_{C/M}$) given in local coordinates by
$$\sum_{\alpha}(-1)^{|\alpha|}D_{\alpha}\left(\frac{\partial L}{\partial x_{i,\alpha}}\right)\circ (j_{\infty}x)(t)=0,$$
where $S=[Ld^{n}t]\in h(\Ac)$ is the local description of the Lagrangian density.

We now define the notion of local variational problem with nice histories. This type
of variational problem can be studied completely by only using geometry of $\Dc$-spaces.
This gives powerful finiteness and biduality results that are necessary to study conceptually
general gauge theories.
\begin{definition}
Let $\pi:C\to M$, $H\subset \uGamma(M,C)$ and $S:H\to \R$ be a Lagrangian
variational problem, and suppose that $S$ is a local functional, i.e., if $\Ac=\Jet(\Oc_{C})$,
there exists $[L\omega]\in h(\Ac):=\Ber_{M}\otimes_{\Dc}\Ac$ and $\Sigma\in H_{c,n}(M)$
such that $S=S_{\Sigma,L\omega}$. The variational problem is called
a \emph{local variational problem with nice histories} if the space of critical points
$T=\{x\in H,\;d_{x}S=0\}$ identifies with the space $\uSol(\Ac/\Ic_{dS})$ of
solutions to the Euler-Lagrange equation.
\end{definition}

The notion of variational problem with nice histories can be explained in simple terms
by looking at the following simple example.
The point is to define $H$ by adding boundary conditions to elements in $\uGamma(M,C)$, so
that the boundary terms of the integration by part, that we do to compute the variation
$d_{x}S$ of the action, vanish.

\begin{example}
Let $\pi:C=\R^3\times [0,1]\to [0,1]=M$, $\Ac=\R[t,x_{0},x_{1},\dots]$ be the corresponding
$\Dc_{M}$-algebra with action of $\partial_{t}$ given by $\partial_{t}x_{i}=x_{i+1}$,
and $S=\frac{1}{2}m(x_{1})^2dt\in h(\Ac)$ be the local action functional for the variational
problem of newtonian mechanics for a free particle in $\R^3$.
The differential of $S:\uGamma(M,C)\to \uR$
at $u:U\to \uGamma(M,C)$ along the vector field $\vec u\in \Theta_{U}$ is given by integrating by part
$$
\langle d_{x}S,\vec u\rangle=\int_{M}\langle -m\partial_{t}^2x,\frac{\partial x}{\partial \vec u}\rangle dt+
\left[\langle \partial_{t}x,\frac{\partial x}{\partial \vec u}\rangle \right]_{0}^1.
$$
The last term of this expression is called the boundary term and we define nice histories
for this variational problem by fixing the starting and ending point of trajectories to annihilate
this boundary term:
$$H=\{x\in \uGamma(M,C),\;x(0)=x_{0},\;x(1)=x_{1}\}$$
for $x_{0}$ and $x_{1}$ some given points in $\R^3$.
In this case, one has
$$T=\{x\in H,\;d_{x}S=0\}\cong \uSol(\Ac/\Ic_{dS})$$
where $\Ic_{dS}$ is the $\Dc$-ideal in $\Ac$ generated by $-mx_{2}$, i.e., by
Newton's differential equation for the motion of a free particle in $\R^3$.
The critical space is thus given by
$$T=\{x\in H,\;\partial_{t}x\textrm{ is constant on }[0,1]\},$$
i.e., the free particle is moving on the line from $x_{0}$ to $x_{1}$ with constant speed.
\end{example}

\begin{definition}
A \emph{general gauge theory} is a local variational problem with nice histories.
\end{definition}

\subsection{Regularity conditions and higher Noether identities}
We now describe regularity properties of gauge theories, basing our exposition
on the article \cite{Fred-histories-and-observables}.
We will moreover use the language of homotopical and derived geometry in the sense
of Toen-Vezzosi \cite{Toen-Vezzosi-HAG-II} to get geometric insights on
the spaces in play in this section
(See \cite{Fred-cours-m2-physique} for an introduction and references).
We denote $\Ac\mapsto Q\Ac$ a cofibrant replacement functor in a given model category.
Recall that all differential graded algebras are fibrant for their standard model structure.

In all this section, we set $\pi:C\to M$, $H\subset \uGamma(M,C)$, $\Ac=\Jet(\Oc_{C})$ and $S\in h(\Ac)$
a gauge theory. The kernel of its insertion map
$$i_{dS}:\Theta_{\Ac}^\ell\to \Ac$$
is called the space $\Nc_{S}$ of Noether identities.
Its right version
$$\Nc_{S}^r=\Ber_{M}\otimes \Nc_{S}\subset \Theta_{\Ac}$$
is called the space of Noether gauge symmetries.

\begin{definition}
The \emph{derived critical space} of a gauge theory is the differential graded $\Ac$-space
$$
\begin{array}{ccccc}
P:=\Spec(\Ac_{P}):& dg-\Ac-\Alg & \to 	  &\SSets\\
    & \Rc		&\mapsto & s\Hom_{dg-Alg_{\Dc}}(\Ac_{P},\Rc).
\end{array}
$$
whose coordinate differential graded algebra is
$$
\Ac_{P}:=
\Sym_{dg}(
[\Theta_{\Ac}^\ell[1]\overset{i_{dS}}{\longrightarrow}\Ac]).
$$
A \emph{non-trivial Noether identity} is a class in $H^1(\Ac_{P},i_{dS})$.
\end{definition}

We refer to Beilinson and Drinfeld's book \cite{Beilinson-Drinfeld-Chiral}
for the following proposition.
\begin{proposition}
The local Lie bracket of vector fields extends naturally to an odd local (so-called Schouten)
Poisson bracket on the dg-$\Ac$-algebra $\Ac_{P}$
of coordinates on the derived critical space.
\end{proposition}

The following corollary explains why we called $\Nc_{S}^r$ the space of Noether gauge
symmetries.
\begin{corollary}
\label{noether-gauge-action}
The natural map
$$\Nc_{S}^r\boxtimes \Nc_{S}^r\to \Delta_{*}\Theta_{\Ac}$$
induced by the local bracket on local vector fields always factors through $\Delta_{*}\Nc_{S}^r$
and the natural map
$$\Nc_{S}^r\boxtimes \Ac^r/\Ic_{S}^r\to \Delta_{*}\Ac^r/\Ic_{S}^r$$
is a local Lie $\Ac$-algebroid action.
\end{corollary}

The trivial Noether identities are those in the image of the natural map
$$\wedge^2\Theta^\ell_{\Ac}\to \Theta_{\Ac},$$
and these usually don't give a finitely generated $\Ac[\Dc]$-module because of
the simple fact that $\Dc\otimes_{\Oc}\Dc$ is not $\Dc$-coherent. This is a very
good reason to consider only non-trivial Noether identities, because these can
usually (i.e., in all the applications we have in mind) be given by a finitely generated
$\Ac[\Dc]$-module.

\begin{definition}
The \emph{proper derived critical space} of a gauge theory is the (derived) space
$$
\begin{array}{ccccc}
\R\Spec(\Ac/\Ic_{S}):& dg-\Ac-\Alg & \to 	  &\SSets\\
    & \Rc		&\mapsto & s\Hom_{dg-Alg_{\Dc}}(\Bc,\Rc).
\end{array}
$$
where $\Bc\overset{\sim}{\longrightarrow} \Ac/\Ic_{dS}$ is a cofibrant resolution of $\Ac/\Ic_{S}$
as a dg-$\Ac$-algebra in degree $0$.
\end{definition}

From the point of view of derived geometry, differential forms on the cofibrant resolution $\Bc$
give a definition of the cotangent complex of the $\Dc$-space morphism
$$i:\uSpec_{\Dc}(\Ac/\Ic_{S})\to \uSpec_{\Dc}(\Ac)$$
of inclusion of critical points of the action functional in the $\Dc$-space of general trajectories.
This notion of cotangent complex gives
a well behaved way to study infinitesimal deformations of the above inclusion map $i$ (see Illusie
\cite{Illusie-cotangent}), even if it is not a smooth morphism (i.e., even if the critical space is
singular).

We will see how to define a cofibrant, so-called Koszul-Tate resolution, that will have good finiteness
properties, by using generating spaces of Noether identities.
These can be defined by adapting Tate's construction
\cite{Tate4} to the local context. We are inspired here by Stasheff's paper \cite{Stasheff2}.
\begin{definition}
\label{generating-space}
A \emph{generating space of Noether identities} is a tuple $(\gfk_{S},\Ac_{n},i_{n})$
composed of
\begin{enumerate}
\item a negatively graded projective $\Ac[\Dc]$-module $\gfk_{S}$,
\item a negatively indexed family $\Ac_{n}$ of dg-$\Ac$-algebras with $\Ac_{0}=\Ac$, and
\item for each $n\leq -1$, an $\Ac[\Dc]$-linear morphism $i_{n}:\gfk_{S}^{n+1}\to Z^n\Ac_{n}$ to
the $n$-cycles of $\Ac_{n}$,
\end{enumerate}
such that if one extends $\gfk_{S}$ by setting $\gfk_{S}^1=\Theta_{\Ac}^\ell$
and if one sets
$$i_{0}=i_{dS}:\Theta_{\Ac}^\ell\to \Ac,$$
\begin{enumerate}
\item one has for all $n\leq 0$ an equality
$$
\Ac_{n-1}=\Sym_{\Ac_{n}}([\gfk_{S}^{n+1}[-n+1]\otimes_{\Ac}\Ac_{n}\overset{i_{n}}{\to} \underset{0}{\Ac_{n}}]),
$$
\item the natural projection map
$$\Ac_{KT}:=\limind \Ac_{n}\to \Ac/\Ic_{S}$$
is a cofibrant resolution, called the \emph{Koszul-Tate algebra},
whose differential is denoted $d_{KT}$.
\end{enumerate}
\end{definition}

We are now able to define the right regularity properties for a given gauge theory.
These finiteness properties are imposed to make the generating space of Noether identities
dualizable as an $\Ac[\Dc]$-module (resp. as a graded $\Ac[\Dc]$-module).
Without any regularity hypothesis, the constructions given by homotopical Poisson reduction of gauge theories,
the so-called derived covariant phase space, don't give $\Ac$-algebras, but only $\R$-algebras,
that are too poorly behaved and infinite dimensional to be of any (even theoretical) use.
We thus don't go through the process of their definition, that is left to the interested reader.

We now recall the language used by physicists (see for example \cite{Henneaux-Teitelboim}) to
describe the situation. This can be useful to relate our constructions to the one
described in physics books.
\begin{definition}
\label{strongly-regular-generating-space}
A gauge theory is called \emph{regular} if there exists a generating space
of Noether identities $\gfk_{S}$ whose components are finitely generated and projective.
It is called \emph{strongly regular} if this regular generating space is a bounded graded
module. Suppose given a regular gauge theory.
Consider the inner dual graded space (well-defined because of the regularity
hypothesis)
$$\gfk_{S}^\circ:=\Homc_{\Ac[\Dc]}(\gfk_{S},\Ac[\Dc])^\ell.$$
\begin{enumerate}
\item The generators of $\Theta_{\Ac}^\ell$ are called \emph{antifields} of the theory.
\item The generators of $\gfk_{S}$ of higher degree are called \emph{antighosts}, or
(non-trivial) \emph{higher Noether identities} of the theory.
\item The generators of the graded $\Ac^r[\Dc^{op}]$-module $\gfk_{S}^r$ are called
(non-trivial) \emph{higher gauge symmetries} of the theory.
\item The generators of the graded $\Ac[\Dc]$-module $\gfk_{S}^\circ$ are called \emph{ghosts} of the theory.
\end{enumerate}
\end{definition}

Remark that the natural map $\gfk_{S}^{0,r}\to \Nc_{S}^r\subset \Theta_{\Ac}$ identifies order
zero gauge symmetries with (densities of) local vector fields that induce tangent vectors
to the $\Dc$-space $\uSpec_{\Dc}(\Ac/\Ic_{S})$ of solutions to the Euler-Lagrange equation.
This explains the denomination of higher gauge symmetries for $\gfk_{S}^r$.

We now define an important invariant of gauge theories, called the Batalin-Vilkovisky bigraded algebra.
This will be used in next section on the derived covariant phase space.
\begin{definition}
Let $\gfk_{S}$ be a regular generating space of the Noether gauge symmetries.
The bigraded $\Ac[\Dc]$-module
$$
\Vc_{BV}:=\left[
\begin{array}{ccccc}
\gfk_{S}[2] &\oplus & \Theta_{\Ac}^\ell[1] &\oplus & 0\\
&&&&\oplus\\
&&&& ^t\gfk_{S}^{\circ}[-1]
\end{array}\right],
$$
where $^t\gfk_{S}^{\circ}$ is the vertical chain graded space associated to $\gfk_{S}^\circ$,
is called the module of \emph{additional fields}. The completed bigraded symmetric algebra
$$\hat{\Ac}_{BV}:=\widehat{\Sym}_{\Ac\textrm{-bigraded}}(\Vc_{BV})$$
is called the \emph{completed Batalin-Vilkovisky algebra} of the given gauge theory.
The corresponding symmetric algebra
$$\Ac_{BV}:=\Sym_{\Ac\textrm{-bigraded}}(\Vc_{BV})$$
is called the \emph{Batalin-Vilkovisky algebra}.
\end{definition}

In practical situations, physicists usually think of ghosts and antifields as sections of
an ordinary graded bundle on spacetime itself (and not only on jet space). This idea
can be formalized by the following.
\begin{definition}
Let $\gfk$ be a regular generating space of Noether symmetries for $S\in h(\Ac)$.
Suppose that all the $\Ac[\Dc]$-modules $\gfk^i$ and $\Theta_{\Ac}^{\ell}$ are locally free on $M$.
A \emph{Batalin-Vilkovisky bundle} is a bigraded vector bundle
$$E_{BV}\to C$$
with an isomorphism of $\Ac[\Dc]$-modules
$$\Ac[\Dc]\otimes_{\Oc_{C}}\Ec_{BV}^*\to \Vc_{BV},$$
where $\Ec_{BV}$ are the sections of $E_{BV}\to C$.
The sections of the graded bundle $E_{BV}\to M$ are called the \emph{fields-antifields variables} of the theory.
\end{definition}

Recall that neither $C\to M$, nor $E_{BV}\to M$ are vector bundles in general.
To illustrate the above general construction by a simple example, suppose that the action $S\in h(\Ac)$
has no non-trivial Noether identities, meaning that for all $k\geq 1$, one has $H^k(\Ac_{P})=0$.
In this case, one gets
$$\Vc_{BV}=\Theta_{\Ac}^\ell[1]$$
and the relative cotangent bundle $E_{BV}:=T^*_{C/M}\to C$ gives a BV bundle because
$$\Theta_{\Ac}^\ell\cong \Ac[\Dc]\otimes_{\Oc_{C}}\Theta_{C/M}.$$
The situation simplifies further if $C\to M$ is a vector bundle because then,
the vertical bundle $VC\subset TC\to M$, given by the kernel of $TC\to \pi^*TM$,
is isomorphic to $C\to M$. Since one has $T^*_{C/M}\cong (VC)^*$, one gets a natural isomorphism
$$E_{BV}\cong C\oplus C^*$$
of bundles over $M$. This linear situation is usually used as a starting point for the definition
of a BV theory (see for example Costello's book \cite{Costello}). Starting from a non-linear
bundle $C\to M$, one can linearize the situation by working with the bundle
$$C^{linear}_{x_{0}}:=x_{0}^*T_{C/M}\to M$$
with $x_{0}:M\to C$ a given solution of the equations of motion (sometimes called the vacuum).

\begin{proposition}
Let $E_{BV}\to C$ be a BV bundle. There is a natural isomorphism of bigraded algebras
$$\Jet(\Oc_{E_{BV}})\overset{\sim}{\longrightarrow} \Ac_{BV}=\Sym_{bigrad}(\Vc_{BV}).$$
\end{proposition}
\begin{proof}
Since $E_{BV}\to C$ is a graded vector bundle concentrated in non-zero degrees, one has
$$\Oc_{E_{BV}}=\Sym_{\Oc_{C}}(\Ec_{BV}^*).$$
The natural map
$$\Ec_{BV}^*\to \Vc_{BV}$$
induces a morphism
$$\Oc_{E_{BV}}=\Sym_{\Oc_{C}}(\Ec_{BV}^*)\to \Ac_{BV}.$$
Since $\Ac_{BV}$ is a $\Dc$-algebra, one gets a natural morphism
$$\Jet(\Oc_{E_{BV}})\to \Ac_{BV}=\Sym_{bigrad}(\Vc_{BV}).$$
Conversely, the natural map $\Ec_{BV}^*\to \Oc_{E_{BV}}$ extends
to an $\Ac[\Dc]$-linear map
$$\Vc_{BV}\to \Jet(\Oc_{E_{BV}}),$$
that gives a morphism
$$\Ac_{BV}\to \Jet(\Oc_{E_{BV}}).$$
The two constructed maps are inverse of each other.
\end{proof}

The main interest of the datum of a BV bundle is that it allows to work with non-local functionals
of the fields and antifields variables. This is important for the effective renormalization of gauge theories,
that involves non-local functionals.

\begin{definition}
Let $E_{BV}\to C$ be a BV bundle. Denote $\A^1(A):=A$ the graded affine space.
The space of \emph{non-local functionals of the fields-antifields}
is defined by
$$\Oc_{BV}:=\uHom(\uGamma(M,E_{BV}),\A^1)$$
of (non-local) functionals on the space of sections of $E_{BV}$.
The image of the natural map
$$h(\Ac_{BV})\cong h(\Jet(\Oc_{E_{BV}}))\longrightarrow \Oc_{BV}$$
is called the space of \emph{local functionals of the fields-antifields}
and denoted $\Oc_{BV}^{qloc}\subset \Oc_{BV}$.
\end{definition}

\subsection{The derived covariant phase space}
\label{covariant-phase-space}
In all this section, we set $\pi:C\to M$, $H\subset \uGamma(M,C)$, $\Ac=\Jet(\Oc_{C})$ and $S\in h(\Ac)$
a gauge theory. Suppose given a strongly regular generating space of Noether symmetries
$\gfk_{S}$ for $S$, in the sense of definitions \ref{generating-space} and
\ref{strongly-regular-generating-space}.

The idea of the BV formalism is to define a (local and odd) Poisson dg-$\Ac$-algebra $(\Ac_{BV},D,\{.,.\})$
whose spectrum $\R\uSpec_{\Dc}(\Ac_{BV},D)$ can be though as a kind of
homotopical space of leaves
$$\R\Spec(\Ac/\Ic_{S})\underset{\R}{/}\Nc_{S}^r$$
of the foliation induced by the action (described in corollary \ref{noether-gauge-action})
of Noether gauge symmetries $\Nc_{S}^r$
on the derived critical space $\R\uSpec_{\Dc}(\Ac/\Ic_{S})$.
It is naturally equipped with a homotopical Poisson structure,
which gives a nice starting point for quantization. From this point of view, the above space is
a wide generalization of the notion extensively used by DeWitt in his covariant approach to
quantum field theory \cite{DeWitt} called the covariant phase space. This explains the title of this section.

We will first define the BV Poisson dg-algebra by using only a generating space
for Noether identities, and explain in more details in the next section how this
relates to the above intuitive statement.

\begin{proposition}
The local Lie bracket and local duality pairings
$$
[.,.]:\Theta_{\Ac}\boxtimes \Theta_{\Ac}\to \Delta_{*}\Theta_{\Ac}\hspace{1cm}\textrm{ and }\hspace{1cm}
\langle.,.\rangle:(\gfk_{S}^n)^r\boxtimes ({\gfk_{S}^n}^\circ)^r\to \Delta_{*}\Ac^r,\;n\geq 0,
$$
induce an odd local Poisson bracket
$$
\{.,.\}:\hat{\Ac}_{BV}^r\boxtimes\hat{\Ac}_{BV}^r\to \Delta_{*}\hat{\Ac}_{BV}^r
$$
called the BV-antibracket on the completed BV algebra 
$$\hat{\Ac}_{BV}=\widehat{\Sym}_{bigrad}\left(\left[\begin{array}{ccccc}
\gfk_{S}[2] &\oplus & \Theta_{\Ac}^\ell[1] &\oplus & \Ac\\
&&&&\oplus\\
&&&& ^t\gfk_{S}^{\circ}[-1]
\end{array}\right]\right)
$$
and on the BV algebra $\Ac_{BV}$.
\end{proposition}

\begin{definition}
Let $\gfk_{S}$ be a regular generating space of Noether identities.
A \emph{formal solution to the classical master equation} is an
$S_{cm}\in h(\hat{\Ac}_{BV})$ such that
\begin{enumerate}
\item the degree $(0,0)$ component of $S_{cm}$ is $S$,
\item the component of $S_{cm}$ that is linear in the ghost variables,
denoted $S_{KT}$, induces the Koszul-Tate differential $d_{KT}=\{S_{KT},.\}$
on antifields of degrees $(k,0)$, and
\item the \emph{classical master equation}
$$\{S_{cm},S_{cm}\}=0$$
(meaning $D^2=0$ for $D=\{S_{cm},.\}$) is fulfilled in $h(\hat{\Ac}_{BV})$.
\end{enumerate}
A \emph{solution to the classical master equation} is a formal solution that comes from
an element in $h(\Ac_{BV})$.
\end{definition}


The main theorem of homological perturbation theory, given in a physical language
in Henneaux-Teitelboim \cite{Henneaux-Teitelboim}, Chapter 17  (DeWitt indices), can be formulated
in our language by the following.

\begin{theorem}
Let $\gfk_{S}$ be a regular generating space of Noether symmetries. 
There exists a formal solution to the corresponding classical master equation,
constructed through an inductive method.
If $\gfk_S$ is further strongly regular and the inductive method ends after finitely many
steps, then there exists a solution to the classical master equation.
\end{theorem}
\begin{proof}
One can attack this theorem conceptually using the general setting of homotopy
transfer for curved local $L_{\infty}$-algebroids (see Schaetz's paper \cite{Schaetz}
for a finite dimensional analog).
We only need to prove the theorem when $\gfk$ 
has all $\gfk^i$ given by free $\Ac[\Dc]$-modules
of finite rank since this is true locally on $M$.
We start by extending $S$ to a generator of the Koszul-Tate differential
$d_{KT}:\Ac_{KT}\to \Ac_{KT}$. Remark that the BV bracket with $S$ on $\Ac_{BV}$ already identifies with
the insertion map
$$\{S,.\}=i_{dS}:\Theta_{\Ac}^\ell\to \Ac.$$
We want to define $S_{KT}:=\sum_{k\geq 0} S_{k}$ with $S_{0}=S$ such that
$$\{S_{KT},.\}=d_{KT}:\Ac_{KT}\to \Ac_{KT}.$$
Let $C_{\alpha_{i}}^*$ be generators of the free $\Ac[\Dc]$-modules $\gfk^i$
and $C^{\alpha_{i}}$ be the dual generators of the free $\Ac[\Dc]$-modules
$(\gfk^i)^\circ$. We suppose further that all these generators correspond to closed elements
for the de Rham differential.
Let $n_{\alpha_{i}}:=d_{KT}(C_{\alpha_{i}}^*)$ in $\Ac_{KT}$.
Then setting $S_{k}=\sum_{\alpha_{k}}n_{\alpha_{k}}C^{\alpha_{k}}$, one gets
$$
\begin{array}{ccc}
\{S_{i},C^*_{\alpha_{i}}\} 	& = & \{n_{\alpha_{i}}C^{\alpha_{i}},C^*_{\alpha_{i}}\}\\
						& = & n_{\alpha_{i}}\\
						& = & d_{KT}(C^*_{\alpha_{i}})
\end{array}
$$
so that $\{S_{KT},.\}$ identifies with $d_{KT}$ on $\Ac_{KT}$.
Now let $m_{\alpha_{j}}$ denote the coordinates of $n_{\alpha_{i}}$
in the basis $C^*_{\alpha_{i}}$, so that
$$n_{\alpha_{i}}=\sum_{j}m_{\alpha_{j}}C^*_{\alpha_{j}}.$$
One gets in these coordinates
$$S_{i}=\sum_{\alpha_{i},\alpha_{j}}C^*_{\alpha_{j}}m_{\alpha_{j}}C^{\alpha_{i}}.$$
The next terms in $S=\sum_{k\geq 0}S_{k}$ are determined by the recursive equation
$$2d_{KT}(S_{k})+D_{k-1}=0$$
where
$D_{k-1}$ is the component of Koszul-Tate degree (i.e., degree in the variables $C^*_{\alpha_{i}}$)
$k-1$ in $\{R_{k-1},R_{k-1}\}$, with
$$R_{k-1}=\sum_{j\leq k-1}S_{j}.$$
These equations have a solution because $D_{k-1}$ is $d_{KT}$-closed, because
of Jacobi's identity for the odd bracket $\{.,.\}$ and since $d_{KT}$ is exact on the Koszul-Tate
components (because it gives, by definition, a resolution of the critical ideal), these are also exact.
If we suppose that the generating space $\gfk_S$ is strongly regular (i.e., bounded) and
the inductive process ends after finitely many steps, one can choose the solution $S$ in $h(\Ac_{BV})$.
\end{proof}

\section*{Acknowledgements}
We refer to the article \cite{Fred-histories-and-observables} for detailed acknowledgements and more
references on this work, that is its direct continuation (with improvements and simplifications).
Special thanks are due to Jim Stasheff for his detailed comments of loc. cit., to the referee, and to my students
at IMPA and Jussieu's master class,  that allowed me to improve both presentation and
results. I also thank the editors, and in particular Urs Schreiber, for giving me the opportunity to publish
in this book.

\bibliographystyle{alpha}
\bibliography{$HOME/Documents/travail/fred}

\end{document}